\documentclass[11pt,a4paper]{article}

\usepackage{vmargin,fancyhdr}   
\usepackage{tikz}
\usepackage{amsthm}

\usepackage[lined, boxed, commentsnumbered, vlined]{algorithm2e}
\usetikzlibrary{decorations.pathreplacing,calc}

\usepackage{amsmath,amssymb}    
\usepackage{verbatim}           

\usepackage[dvips]{thumbpdf}

\setpapersize{USletter}
\setmarginsrb{.75in}{.5in}        
            {.75in}{.5in}        
            {.25in}{.25in}     
            {.25in}{.5in}      

\newtheorem{theorem}{Theorem}[section]

\newtheorem{definition}[theorem]{Definition}

\newtheorem{lemma}[theorem]{Lemma}
\newtheorem{fact}[theorem]{Fact}

\newcommand{\todo}[1]{{\color{red}\textbf{TODO:}#1}}

\newcommand{\poly}[1]{\textbf{poly}(#1)}

\newcommand{\ignore}[1]{{}}
\newcommand{\optmatching}{\mathcal{M}}
\newcommand{\optmatchinglevel}[1]{\mathcal{M}(#1)}
\newcommand{\matchinglevel}[1]{M_{#1}}
\newcommand{\optmatchingsize}[1]{|\mathcal{M}(#1)|}
\newcommand{\optmatchingweight}[1]{{w(\mathcal{M}(#1))}}
\newcommand{\combinedmatching}{\hat{M}}

\newcommand{\optmatchingcopy}[1]{\mathcal{M}^{#1}}
\newcommand{\matchinglevelcopy}[2]{M^{#2}_{#1}}
\newcommand{\optmatchinglevelcopy}[2]{\mathcal{M}^{#2}_{#1}}

\newcommand{\combinedmatchingcopy}[1]{\hat{M}^{#1}}

\newcommand{\vcover}{V_{cover}}
\newcommand{\maxweight}{N}

\newcommand{\defeq}{:\stackrel{\text{def}}{=}}

\begin{document}

\title{Fully Dynamic $(1+\epsilon)$-Approximate Matchings}

\author{
Manoj Gupta\\
IIT Delhi\\
\texttt{gmanoj@cse.iitd.ac.in}
\and
Richard Peng\\
CMU\\
\texttt{yangp@cs.cmu.edu}
}

\maketitle

\begin{abstract}
We present the first data structures that maintain
near optimal maximum cardinality and maximum weighted matchings
on sparse graphs in sublinear time per update.
Our main result is a data structure that maintains a
$(1+\epsilon)$ approximation of maximum matching under edge insertions/deletions
in worst case $O(\sqrt{m}\epsilon^{-2})$ time per update.
This improves the $3/2$ approximation given
in [Neiman, Solomon, ~STOC~2013] which runs in similar time.
The result is based on two ideas.
The first is to re-run a static algorithm after a chosen number
of updates to ensure approximation guarantees.
The second is to judiciously trim the graph to a smaller equivalent
one whenever possible.

We also study extensions of our approach to the weighted
setting, and combine it with known frameworks to obtain arbitrary
approximation ratios.
For a constant $\epsilon$ and 
for graphs with edge weights between $1$ and $\maxweight$,
we design an algorithm that maintains an $(1+\epsilon)$-approximate
maximum weighted matching in $O( \sqrt{m} \log \maxweight)$ time per update.
The only previous result for maintaining weighted matchings
on dynamic graphs has an approximation ratio of 4.9108,
and was shown in [Anand, Baswana, Gupta, Sen, ~FSTTCS~2012, ~arXiv~2012].
\end{abstract}

\section{Introduction}
\label{sec:intro}

The problem of computing maximum or near-maximum matchings
in a graph has played a central role in the study of
combinatorial optimization\cite{LovaszP86,PapadimitriouS82}.
A matching is a set of vertex-disjoint edges in a graph, and two variants
of the problem are finding the maximum cardinality matching in 
an unweighted graph, and finding the
matching of maximum weight in a weighted graph.
The problem is appealing for several reasons: it has a simple description;
matchings sometimes need to be improved by highly non-local steps;
and certifying the optimality of a matching yields a surprising
amount of structural information about a graph.
On static graphs, the current best
algorithms for maximum cardinality matching run in $O(m \sqrt{n})$ time,
on bipartite graph by Hopcroft and Karp \cite{hopcroft1971n5},
and on general graph by Micali and Vazirani \cite{micali1980v}.
In the weighted case, algorithms with similar running times were
given by Gabow and Tarjan~\cite{gabow1991faster},
and by Duan et al.~\cite{DuanPS11}.

A natural question from a data structure perspective is
whether on a dynamically changing graph the solution to
an optimization problem can be maintained faster than recomputing
it from scratch after each update.
For maximum cardinality matching, an $O(m)$ time algorithm follows
by executing one phase of the static algorithm described by Tarjan\cite{Tarjan83}.
For dense graphs,  a faster running time of $O(n^{1.495})$ has been
shown by Sankowski \cite{sankowski2007faster}, and to date this is
the only known result that gives sublinear time per update.
For trees, Gupta and Sharma\cite{gupta2009log} gave an algorithm
based on top trees that takes $O(\log n)$ time per update.

On static graphs, a nearly-optimal matching can be computed
much faster than finding the optimum matching.
So it stands to reason that the same should apply in the dynamic case.
Ivkovi\'c  and Llyod\cite{ivkovi1994fully} gave the first
result in this direction:
an algorithm that maintains a maximal matching with
$O((n + m)^{0.7072})$ update time.
Recently there has been a growing interest in designing 
efficient dynamic algorithms for approximate matching.
Onak and Rubinfeld designed a randomized algorithm that
maintains a $c$-approximation of maximum matching 
in $O(\log^2 n)$  update time \cite{onak2010maintaining},
where $c$ is a large unspecified constant.
Baswana, Gupta and Sen\cite{baswana2011fully} showed that
maximal matching, which is a $2$-approximation of maximum matching, 
can be maintained in a dynamic graph
in amortized $O(\log n)$ update time with high probability. 
Subsequently, Anand et al.~\cite{AnandBGS12,AnandBGS-Arxiv12} extended
this work to the weighted case, and showed how to maintain a matching
with weight that is {\em expected} to be at least
$1/4.9108\approx0.2036$ of the optimum.

These results show that a large matching can be maintained
very efficiently in dynamic graphs, but leave open
the question of maintaining a matching closer to the optimum matching.
Recently, Neiman and Solomon\cite{neiman12deterministic} showed that a matching of size at least
$2/3$ of the size of optimum matching
can be maintained in $O(\sqrt{m})$ time per update in general graphs
, as well as $O(\log{n}/ \log\log{n})$ time
per update on bounded arboricity graphs.
A similar result of maintaining 3/2-approximate matchings was 
obtained independently by Anand \cite{Anand12}.
This leads to the following question:  Can we maintain a matching 
close to maximum matching (say $(1+\epsilon)$-approximate matching) in 
a dynamic weighted or unweighted graph?
We answer this question in affirmative  by designing the 
first data structure that maintains
arbitrary quality approximate max-cardinality and max-weighted
matching in sublinear time on sparse graph.

Our algorithm differs significantly from previous ones
in that we do not maintain strict invariants.
Baswana et al. \cite{baswana2011fully} maintained a maximal
matching, which ensures no edge has both endpoints unmatched;
and the $3/2$-approximate algorithm designed by Neiman and
Solomon\cite{neiman12deterministic} 
remove all length three augmenting paths in the graph at each update step.
Our approach makes crucial use of the fact that the optimization
objectives involving matching is {\em stable}.
That is, a single update can only change the value of the optimum matching by $1$.
So if we find a matching close to maximum matching at some update step,
it remains close to maximum even after several updates to the graph.
In case the current matching ceases to be a good approximation 
of the maximum matching, we then re-run the static algorithm to
get a matching that is close to optimum.
This approach of re-running a expensive routine occasionally is a
common technique in dynamic graph data structures
\cite{HenzingerK99,HolmLT01,BaswanaKS12}.
It is particularly powerful for approximating matchings since the stability
property gives us freedom in choosing when to re-run the static algorithms.
But re-running static algorithm occasionally works well
when the maximum matching in the graph is large.
To deal with graphs having small maximum matching, we introduce the concept of 
{\em core subgraph} which is the central concept of 
our paper. A {\em core subgraph} is a subgraph of a graph 
having the following two properties:
Its size is considerably 
smaller than the entire graph. Secondly, the size of maximum matching
in {\em core subgraph} is same as the size of maximum matching in
the entire graph. We will crucially use these two properties in designing 
a dynamic algorithm for approximate matching.
A detailed description of our algorithm, as well as other components of
our data structure are presented in Section~\ref{sec:algo}
~and~Appendix~\ref{sec:improvementsdetails}.
The main result for approximating the maximum cardinality
matching can be stated as follows:

\begin{theorem}
\label{thm:maxcardinality}
For any constant $\epsilon<1/2$, there exists an algorithm 
which maintains a $(1+\epsilon)$-approximate matching
in an unweighted dynamic graph in $O( \sqrt{m} )$ worst case update time.
\end{theorem}
 
It can be argued that the {\em stability} property of matchings
that we rely on is rare among optimization problems.
For most other problems like shortest paths and minimum spanning tree, 
there exist updates that require
immediate changes in the approximate solution maintained.
For matchings, such updates exists in the weighted version,
where the objective is the sum of weight over edges in the matching.
Direct extensions of our approach have linear dependencies on
$\maxweight$ in update time, where $\maxweight$
is  the maximum weight of an edge.
This dependency can in fact be viewed as a quantitative
measurement of the decrease in stability as we allow larger weights.

As a result, we investigate rounding/bucketing based approaches
which have logarithmic dependency on $\maxweight$ in
Section~\ref{sec:weightedscaling}.
This was first studied for maintaining dynamic matchings by
Anand et al. \cite{AnandBGS12}, and they used dynamic maximal
matchings as a subroutine in their algorithm.
Directly substituting our result for maximum cardinality matching
leads to immediate improvements in the approximation ratio which is
the second result in this paper.
\begin{theorem}
\label{thm:weightedthree}
For any constant $\epsilon <1/2$, there exists an algorithm that maintains 
$(3+\epsilon)$-approximate maximum weighted matching in 
a graph where edges have weights between $[1, \maxweight]$ in 
$O(\sqrt{m} \log \maxweight)$ worst case update time.
\end{theorem}

Our $(3 + \epsilon)$-approximation algorithm is derived from known
schemes which bucket edges based on their weights.
The rounding scheme we use in this algorithm is 
based on algorithm designed by Anand et al. \cite{AnandBGS12}.
It is not clear whether any extension of this bucketing scheme will lead to
a $(1+\epsilon)$-approximate matching.
To do this, we devise a new rounding scheme which obtains arbitrarily
good approximations of maximum matching, albeit at the cost of a much
higher dependency on $1 / \epsilon$ in the running time.

\begin{theorem}
\label{thm:weightedarbitrary}
For any constant $\epsilon<1/2$, there exists an algorithm that maintains 
$(1+\epsilon)$-approximate maximum weighted matching in
a graph where edges have weights between $[1, \maxweight]$
in $O(\sqrt{m} \log \maxweight)$ worst case update time.
\end{theorem}


As with the algorithm by Neiman and Solomon~\cite{neiman12deterministic},
our algorithms  are deterministic and the update time guaranteed by them 
is worst case.
However, for simplicity in our presentations we will often start
by describing the simpler amortized variants.




\section{Preliminaries}
\label{sec:prelim}

We start by stating the notations that we will use,
and reviewing some well-known results on matchings.
An undirected graph is represented by $G = (V, E)$,
where $V$ represents the set
of vertices and $E$ represents the set of edges in the graph.
We will use $n$ to denote the number of vertices $|V|$, and $m$
to denote the number of edges $|E|$. 

A \emph{matching} in a graph is a set of independent edges 
in the graph.
Specifically, a subset of edges, $M \subseteq E$ is a matching
if no vertex of the graph is incident on more than one edge in $M$.
A vertex is called \emph{unmatched} if it is not incident on any edge in $M$,
otherwise it is \emph{matched}.
Similarly, an edge is called \emph{matched} if it is in $M$
or \emph{free} otherwise. A vertex cover is
a set of vertices in a graph such that each edge has at least one 
of its endpoint in the vertex cover.

The maximum cardinality matching(MCM) in a graph is the
matching of maximum size. 
Similarly, given a set of weights $w: E \rightarrow [1, \maxweight]$,
we can denote the weight of a matching $M$ as
$w(M) = \sum_{e \in M} w(e)$.
The maximum weight matching(MWM) in a graph is in turn
the matching of maximum weight.
We will use $\optmatching$ to denote a optimum matching
for either of these two objectives depending on context.

For measuring the quality of approximate matching, we will use
the notation of $\alpha$-approximation, which indicates that the objective
(either cardinality or weight) given by the current solution is at least
$1/\alpha$ of the optimum.
Specifically, a matching $M$ is called $\alpha$-MCM
if $|M| \ge \frac{1}{\alpha} |MCM|$, and $\alpha$-MWM
if $w(M) \ge \frac{1}{\alpha} |MWM|$.

%
%
Finding or approximating MCMs and MWMs
in the static setting have been intensely studied.
Nearly linear time algorithms have been developed for finding
$(1+\epsilon)$ approximations, and we will make crucial use
of these algorithms in our data structure.
For maximum cardinality matching, such an algorithm for bipartite
graph was introduced by Hopcroft and Karp\cite{hopcroft1971n5},
and extended to general graphs by Micali and Vazirani\cite{micali1980v,Vazirani12}.

\begin{lemma}
\label{lem:staticmcm}
There exists an algorithm \textsc{ApproxMCM} that when given a graph $G$ with
$m$ edges along with a parameter $\epsilon < 1$,
return an $(1 + \epsilon)$-MCM in $O(m \epsilon^{-1})$ time.
\end{lemma}

For approximate MWM, there has been some recent progress.
Duan et al.\cite{DuanP10, DuanPS11} designed an algorithm that 
find a $(1+\epsilon)$ approximate maximum weighted matching 
 in  $O(m \epsilon^{-1} \log(\epsilon^{-1}))$ time.

\begin{lemma}\cite{DuanP10, DuanPS11}
\label{lem:staticmwm}
There exists an algorithm \textsc{ApproxMWM} that when given a graph $G$ with
$m$ edges along with a parameter $\epsilon < 1$,
return an $(1 + \epsilon)$-MWM in $O(m \epsilon^{-1} \log(\epsilon^{-1}))$ time.
\end{lemma}

All logarithms in this paper are with base 2 unless mentioned otherwise.

\ignore{\todo{Can doubly linked lists be used for these?}
The algorithm requires maintaining several sets
of vertices. We can use any standard balanced tree\cite{adelsonvelskii1963algorithm} for maintaining these sets.
Using these balanced trees we can insert to or delete from the set in worst case 
$O(\log n)$ update time as there are at most $n$ elements in the set of vertices.
We would also require enumerating all the elements of a set. Using balanced trees
takes time proportional to the size of the set.
After every update we manipulate $O(\sqrt m)$ sets and
hence the time complexity of every update becomes $O(\sqrt m \log n)$.
}

\section{$(1 + \epsilon)$-MCMs Using Lazy Updates}
\label{sec:algo}

\subsection{Overview}
To maintain approximate matching, we exploit the {\em stability}
of the matching and use the static algorithm 
for matching \textsc{ApproxMCM} periodically.
Our starting point is the observation that the size of maximum
matching changes by at most $1$ per update.
This means that if we have a large matching that's close to
the maximum, it will remain close to maximum matching over a large number of updates.
So we use the following approach: Find a matching at certain update step
and wait for certain number of updates till the matching is a 
good approximation of maximum matching.
This approach works well if the maximum matching is 
itself large to begin with.
But if the maximum matching itself is small,
we still need to run the static algorithm many times.

To overcome this, we show that instead of finding a maximum matching
on the entire graph, we can use a small {\em special} subgraph 
such that the size of maximum matching in this subgraph is 
same as the size of maximum matching in the entire graph.
We call this subgraph a {\em core subgraph}, and it is the
central idea of our $(1+\epsilon)$ approximate algorithm.
As this subgraph is considerably smaller, the time needed
to find a maximum matching on it is considerably less.
We will show that this {\em core subgraph} can be formed
using the vertex cover of the entire graph.
Specifically, we take the vertex-induced subgraph formed by
the cover, along with some {\em special} chosen edges
out of vertices belonging to the cover.

But this leads to another question: How 
do we maintain a vertex cover in a dynamic graph? 
For this, we can use the algorithm of Neiman and 
Solomon \cite{neiman12deterministic}.
One of the invariants in this algorithm is that there are
no edges between unmatched vertices, which means the set of matched
vertices form a $2$-approximate minimum vertex cover.
Therefore reporting these vertices suffices for a vertex cover
at any update step.
However, note that our dependence on the above algorithm is not critical.
Specifically, we design another simple algorithm which does not depend
on the algorithm of Neiman and Solomon\cite{neiman12deterministic} 
for finding the {\em core subgraph}.
A description of this, as well as modifications for handling
edges with weights in a small range, and obtaining worst case bounds
are in Section~\ref{subsec:improvements}, with details
deferred to Appendix~\ref{sec:improvementsdetails}.


\subsection{Algorithm}
We start with some notations that we will use in this section.
We number the updates from $1$ to $t$ and use the following
notations:
\begin{itemize}
\item $G(i)$: The graph after the i\textsuperscript{th} update.
\item $M(i)$: A matching computed on $G(i)$
\item $M(i \setminus j)$: Let $del_M(i,j)$ denote the set of all edges  in $M(i)$
that  are deleted from the graph 
between update steps $i$ and $j$. We define $M( i \setminus j)$ to be
$M(i) \setminus del_M(i,j)$, i.e., $M(i \setminus j)$ consists of all the edges 
in the matching $M(i)$ that are not deleted between update step $i$ and $j$.

\end{itemize}

Also, we will use $\optmatchinglevel{i}$ to denote the optimal matching at step $i$.
The approximation guarantees of $M(i \setminus j)$ is as follows:

\begin{lemma}
\label{lem:stable}
If $\epsilon, \epsilon' \leq 1/2$ and
$M(i)$ is an $(1 + \epsilon)$-MCM in $G(i)$,
then for $j \leq i + \epsilon' |M(i)|$ 
$M(i \setminus j)$ is an $(1 + 2 \epsilon + 2 \epsilon')$-MCM in $G(j)$
\end{lemma}

\begin{proof}
Suppose there were $k_{ins}$ insertions and $k_{del}$ deletions in
the $k=\epsilon' |M(i)|$ updates between updates $i$ and $j$.
The assumption about $M(i)$ implies that $\optmatchingsize{i}
\leq (1 + \epsilon) |M(i)|$.
Since each insert can increase the size of the maximum
matching by $1$, we have
$\optmatchingsize{j} \leq \optmatchingsize{i} + k_{ins}$.
Also, each deletion can remove at most one edge from $M(i)$,
so $|M(i \setminus j)| \geq |M| - k_{del}$.
The approximation ratio is then at most:
\begin{align*}
\frac{\optmatchingsize{j}}{|M(i \setminus j)|}
& \leq \frac{(1 + \epsilon) |M(i)| + k_{ins}}{|M(i)| - k_{del}}  \nonumber\\
& =  1 + \frac {\epsilon |M(i)| + k }{|M(i)| - k_{del}}            \nonumber \\
& \leq  1 + \frac{\epsilon |M(i)| + \epsilon' |M(i)|} {1/2 |M(i)|}
\qquad \text{Since $k_{del} \leq \epsilon' |M(i)| \leq 1/2 |M(i)|$} \nonumber \\
& \leq  1 + 2 \epsilon + 2 \epsilon'
\end{align*}
\end{proof}

This fact has immediate algorithmic consequences for situations
where the maximum matching is large.
Suppose we computed an $(1 + \epsilon / 4)$-MCM for $G(i)$, $M(i)$,
then $M(i \setminus j)$ is $(1+\epsilon)$ approximate maximum matching 
as long as $j \leq i + \epsilon |M(i)| / 4$.
The $O(m \epsilon^{-1})$ cost of the call to \textsc{ApproxMCM}
(given by Lemma~\ref{lem:staticmcm}) can then be charged
to the next $\epsilon |M(i)| / 4$ updates, giving
$O(\frac{m}{|M(i)|} \epsilon^{-2})$ time per update.
When $|M(i)|$ is large, this cost is fairly small.
On the other hand, when $|M(i)|$ is of constant size,
this approach will make a
call to \textsc{ApproxMCM} almost every update.

\ignore{
\todo{My guide has ask me to remove this paragraph}
We reduce the cost of these calls when $|M(i)|$ is small by showing
that \textsc{ApproxMCM} can be ran on only a part of the graph.
For this, we turn to a key idea from the work by Neiman et al.
\cite{neiman12deterministic},
which is to consider a small subset of the vertices as special.
Our choice of these vertices stems from a crucial step in the
data structure by Onak and Rubinfield \cite{onak2010maintaining},
namely that small maximum matchings also implies small vertex covers.
Specifically, there exist a small set of vertices $\vcover$ such that
all edges have at least one end point in $\vcover$.
}


For small size matching, we introduce the concept of {\em core subgraph}.
As mentioned previously, {\em core subgraph} can be found by using 
a vertex cover $G$.

\begin{definition}
\label{def:core}
Given a graph $G$ and a vertex cover $\vcover$, a  core subgraph
$G'$ consists of:
\begin{itemize}
\item All edges between vertices in $\vcover$
\item For each vertex $v \in \vcover$, the $|\vcover| + 1$ edges
of maximum weight of $v$ to vertices in $V \setminus \vcover$.
In case of an unweighted graph, these edges can be chosen arbitrarily.
\end{itemize}

\end{definition}

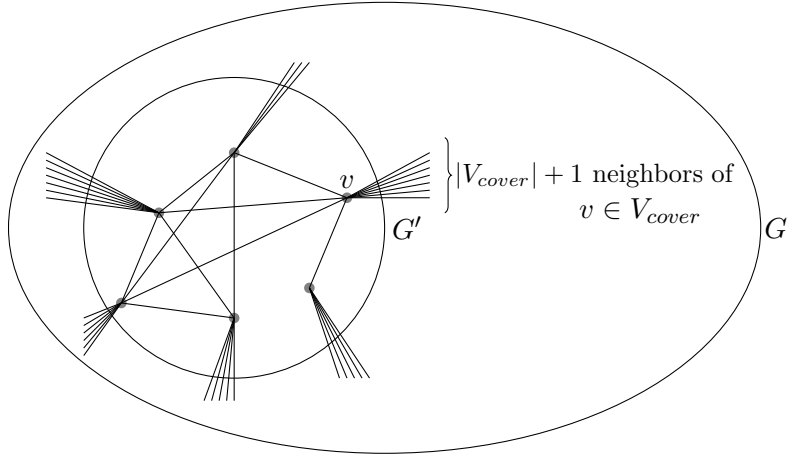
\begin{figure}
 \begin{tikzpicture}
 \draw (0,0) node {};
\draw (8,0) ellipse (5cm and 3cm);
 \draw (13.2,0) node {$G$};
  \draw (8.3,0) node {$G'$};
 \draw (6,0) circle (2cm);
 
 \coordinate (A1) at (5,.2);
\fill [black,opacity=.5] (A1) circle (2pt); 

 \coordinate (A2) at (6,1);
\fill [black,opacity=.5] (A2) circle (2pt); 

 \coordinate (A3) at (4.5,-1);
\fill [black,opacity=.5] (A3) circle (2pt);

 \coordinate (A4) at (6,-1.2);
\fill [black,opacity=.5] (A4) circle (2pt);

 \coordinate (A5) at (7,-.8);
\fill [black,opacity=.5] (A5) circle (2pt);

 \coordinate (A6) at (7.5,.4);
\fill [black,opacity=.5] (A6) circle (2pt);
 \draw (A6)[above] node {$v$};

\draw (A1) -- (A2);
\draw (A1) -- (A3);
\draw (A1) -- (A4);
\draw (A1) -- (A6);
\draw (A2) -- (A3);
\draw (A2) -- (A4);
\draw (A2) -- (A6);
\draw (A3) -- (A4);
\draw (A3) -- (A6);
\draw (A5) -- (A6);

\draw (A1) -- (3.5,.4);
\draw (A1) -- (3.5,.5);
\draw (A1) -- (3.5,.6);
\draw (A1) -- (3.5,.7);
\draw (A1) -- (3.5,.8);
\draw (A1) -- (3.5,.9);
\draw (A1) -- (3.5,1);

\draw (A2) -- (7,2.2);
\draw (A2) -- (6.9,2.2);
\draw (A2) -- (6.8,2.2);

\draw (A3) -- (4,-1.5);
\draw (A3) -- (4,-1.2);
\draw (A3) -- (4,-1.3);
\draw (A3) -- (4,-1.4);
\draw (A3) -- (4,-1.6);
\draw (A3) -- (4,-1.7);

\draw (A4) -- (5.6,-2.3);
\draw (A4) -- (5.7,-2.3);
\draw (A4) -- (5.8,-2.3);
\draw (A4) -- (5.9,-2.3);
\draw (A4) -- (6,-2.3);

\draw (A5) -- (7.7,-2);
\draw (A5) -- (7.4,-2);
\draw (A5) -- (7.5,-2);
\draw (A5) -- (7.6,-2);
\draw (A5) -- (7.8,-2);

\draw (A6) -- (8.6, 1);
\draw (A6) -- (8.6, .9);
\draw (A6) -- (8.6, .8);
\draw (A6) -- (8.6, .7);
\draw (A6) -- (8.6, .6);
\draw (A6) -- (8.6, .5);
\draw (A6) -- (8.6, .4);

\draw [decorate,decoration={brace,mirror}]
(8.8,.2) -- (8.8,1.2)  node [midway,right] {\small{$|\vcover |+1$ neighbors of} } ;
 \draw (11.4,.23)node {$v \in \vcover$ };
\end{tikzpicture}
\caption{An example showing the {\em core subgraph} $G'$ of G. All the vertices in 
the inner circle form a vertex cover $\vcover$ of G.
The {\em core subgraph} contains all the edges induced by the vertices in $\vcover$ plus atmost $|\vcover|+1$ edges from 
each vertex $v \in \vcover$ whose other endpoint is not in the vertex cover}
\label{fig:coresubgraph}
\end{figure}

An illustration of a {\em core subgraph} is shown in Figure
\ref{fig:coresubgraph}. It can be used algorithmically as follows.

\begin{lemma}
\label{lem:smallcover}
Let $G'$ be a core subgraph of $G$ formed using a vertex cover
$\vcover \subseteq V$.
If $M'$ is a $(1 + \epsilon)$-MCM in $G'$, then it's
also a $(1 + \epsilon)$-MCM in $G$.
\end{lemma}

\begin{proof}
We first show that the size of the maximum matching in $G$
is the same as the size of the maximum matching in $G'$.
Among all maximum matchings in $G$, let
$\optmatching$ be one that uses the maximum number of edges in $E(G')$.
For the sake of contradiction, suppose $\optmatching$ 
contains an edge $(u, v)$ in $E(G) \setminus E(G')$.
Since $\vcover$ is a vertex cover, one of $u$ or $v$ is in $\vcover$,
without loss of generality assume it's $u$.
By the construction rule, for $(u, v)$ to not be included in  $G'$,
there exists $|\vcover| + 1$ neighbors of $u$ in $V \setminus \vcover$
that are in $G'$, let them be $N_{V \setminus \vcover}(u)$.
As the maximum matching in $G$ has size at most $|\vcover|$
and there are no edges with both endpoints in $V \setminus \vcover$,
at most $|\vcover|$ vertices in $N_{V \setminus \vcover}(u)$ can be matched.
Therefore there exists an unmatched vertex $x$ in
$N_{V \setminus \vcover}(u)$.
Substituting $(u, v)$ with $(u, x)$ gives a maximum matching
that uses one more edge in the $G'$, giving a contradiction.

Combining this with the fact that $E(G') \subseteq E(G)$ implies that the size of the
maximum matchings in $G$ and $G'$ are the same.
Therefore any $(1 + \epsilon)$-MCM in $G'$ is also a $(1 + \epsilon)$-MCM in $G$. 
\end{proof}

As mentioned previously, we can find $\vcover$ in the graph by using
the algorithm of Neiman and Solomon \cite{neiman12deterministic}.
Their algorithm maintains 3/2 approximate matching in $O( \sqrt{m} )$ 
update time in the worst case which is less than the bound we are claiming.
Whenever we need a vertex cover, we can report all the matched vertices
in the 3/2 approximate matching. From now on we will assume an oracle 
access to the vertex cover at any update step.
A more detailed treatment of maintaining a small cover can be found
in Appendix~\ref{subsec:vcover}.

Any vertex cover $\vcover$ in graph $G(i)$ formed out of a valid matching has the 
following property: $|\vcover| \leq 2 \optmatchingsize{i}$. This is because the 
size of any valid matching is always less than the maximum matching size $\optmatchingsize{i}$.
Therefore when $\optmatchingsize{i}$ is small,
we only need to run the static algorithm given by
Lemma~\ref{lem:staticmcm} on a {\em core subgraph}
$G'(i)$ of $G(i)$.
We can construct this graph
in $O(|\vcover|^2)( = O( \optmatchingsize{i}^2))$ time by examining up to
$O(|\vcover|)$ neighbors of each vertex in $\vcover$.
Using Lemma~\ref{lem:staticmcm}, we can find a $(1+\epsilon)$
approximate matching in this graph in $O( \optmatchingsize{i}^2 \epsilon^{-1})$
time.
Furthermore, Lemma~\ref{lem:stable} allows us to charge
this $O(\optmatchingsize{i}^2 \epsilon^{-1})$ time to the next
$\epsilon \optmatchingsize{i}/ 4$ updates.
Therefore, cost charged per update can be bounded by
$O(\optmatchingsize{i} \epsilon^{-2})$, which is small for small values of $\optmatchingsize{i}$.
Our data structure maintains the following global states:
\begin{enumerate}
	\item A matching $M$.
	\item A counter $t$ indicating the number of updates until we make the next call to \textsc{ApproxMCM}
	\item A vertex cover $\vcover$  (Using the algorithm of Neiman and Solomon \cite{neiman12deterministic})
\end{enumerate}

Upon initialization $M$ is obtained by running the static
algorithm on $G$, or can be empty if $G$ starts empty.
$t$ can be initialized to $\epsilon / 4 |M|$.
Since we handle insertions and deletions in almost symmetrical
ways, we present them as a single routine $\textsc{Update}$,
shown in Figure \ref{fig:lazysimple}

\begin{figure}[ht]
\begin{procedure}[H]
	\nl \If{Update is a deletion and $(u,v) \in M$}{
		\nl Remove $(u, v)$ from $M$\;
	}
	\nl $t \leftarrow t - 1$ \;
	\nl \If{$t \leq 0$}{
		\nl Construct a {\em core subgraph $G'$} of the current graph\;
		\nl $M \leftarrow \textsc{ApproxMCM}(G', 1 + \epsilon / 4)$ \;
		\nl $t \leftarrow \epsilon / 4 |M|$ \;
	}
\caption{Update(u, v)}
\end{procedure}
\caption{Lazy update algorithm for maintaining $(1 + \epsilon)$-MCMs}
\label{fig:lazysimple}
\end{figure}

The bounds of this routine is as follows:

\begin{theorem}
\label{thm:lazy}
The matching $M$ is an $(1 + \epsilon)$-MCM over all updates.
Furthermore, the amortized cost per update is
$O(\sqrt{m} \epsilon^{-2})$.
\end{theorem}

\begin{proof}
Let the current update be at time $j$, and the matching
$M$ that we maintained was computed in iteration $i < j$.
So at update step $i$, the matching $M$ is same as $M(i)$ 
and at update step $j$, it  is $M(i \setminus j)$.
If $t > 0$, then since $t$ was initialized to $\epsilon / 4 |M(i)|$, we
have $j - i \leq \epsilon / 4 |M(i)|$.
The guarantees for $M(i \setminus j)$ follows from Lemma \ref{lem:stable}
with $\epsilon \leftarrow \epsilon / 4$ and $\epsilon' \leftarrow \epsilon /4$.

We now turn our attention to running time.
Consider a call to \textsc{ApproxMCM} made at update $i$.
Assume that $\epsilon |M(i)| \ge 1$.
We have seen that there exists a {\em core subgraph} $G'(i)$ such that the number of edges
$|E(G'(i))|$ can be bounded by $O(\min \{ m, \optmatchingsize{i}^2 \} )$.
Since $M(i)$ is a $(1+\epsilon/4)$ approximate matching, $(1+\epsilon/4)|M(i)| \ge  \optmatchingsize{i}$.
So, the size of $E(G'(i))$ is $O(\min \{ m, |M(i)|^2 \} )$.
Moreover, the cost of finding the matching( in \textsc{ApproxMCM}) in the graph can be at most 
$O(\min \{ m, |M(i)|^2 \} \epsilon^{-1})$.
This cost can be charged to the $\epsilon|M(i)|/4$ updates
starting at update $i$, implying the following amortized cost
per update:
\begin{align*}
\frac{O(\min \{ m, |M(i)|^2 \} \epsilon^{-1})}{\frac{\epsilon}{4} |M(i)|}
= O \left( \min \left\{ \frac{m}{|M(i)|}, |M(i)| \right\} \epsilon^{-2} \right)
\end{align*}
If $|M(i)| \geq \sqrt{m}$, the first term inside $\min$ is at most
$\sqrt{m}$, otherwise the second is at most $\sqrt{m}$.
Combining these two cases gives our desired bound.

Now we take a look at some corner cases to complete the proof. We assumed that 
the cost of finding the matching at level $i$ can be charged to next $\epsilon | M(i)|/4$
updates. This is true except for last call to \textsc{ApproxMCM}. The number of updates 
after this last call can be less than $\epsilon |M(i)|/4$. This cost can be amortized to all
the updates. Since the number of updates is at least $m$, the total cost 
charged to each update step is $O( \epsilon^{-1})$.

The other case is when $\epsilon |M(i) | < 1$. This implies that $G'(i)$ has size 
at most $O( \epsilon^{-2})$ and finding a matching in such a graph takes time $O( \epsilon^{-3})$.
For any constant $\epsilon$, this bound is $O(\sqrt{m} \epsilon^{-2})$
and can be charged to update step itself.

So the amortized cost charged to any update step is at most $O( \sqrt{m} \epsilon^{-2})$.
\end{proof}

\subsection{Improvements, Worst-Case Bound, and Weights}
\label{subsec:improvements}

Several improvements can be made to the simpler version of our
algorithm described above.
We state the main statements here, and more details on these
modifications can be found in Appendix~\ref{sec:improvementsdetails}.

First, note that we depend on the algorithm of Neiman 
and Solomon \cite{neiman12deterministic} to maintain approximate vertex 
cover. Instead of using their algorithm, we design another
simple dynamic algorithm which maintains approximate vertex cover.
This algorithm is similar in spirit as our 
approximate matching algorithm, i.e, we use the property that
vertex cover are also {\em stable}  and a single update to the graph
can change the vertex cover by 1. Using the techniques similar to the 
one presented in the previous section, we design an algorithm in Appendix~\ref{subsec:vcover}
which take $O(\sqrt{m})$ update time in the worst case to maintain
approximate vertex cover. 

Note that  \textsc{ApproxMCM} may take $O(m \epsilon^{-1})$ time in the worst case.
So our algorithm in the previous section had an {\em amortized}
running time of $O(\sqrt{m} \epsilon^{-2})$ per update.
We show that we can maintain approximate matching in worst case
$O( \sqrt{m} \epsilon^{-2})$ update time.
Specifically, we show in Appendix~\ref{subsec:worst} that computation 
cost of $O(m \epsilon^{-1})$ time in \textsc{ApproxMCM}
can  be distributed across
a number of updates. 

Furthermore, our ideas of maximum cardinality matching
can also be adapted to maximum
weighted matchings.
This extension is natural because maximum cardinality matchings
is a special case where all edges have weight $1$.
A closer examination of the proofs of Lemma~\ref{lem:stable}
shows that when all edge weights are in the range $[1, \maxweight]$,
the stability properties only degrade by a factor of $\maxweight$.
In Appendix~\ref{subsec:weights}, we present the following result:

\begin{theorem}
\label{thm:weightedsimple}
For any constant $\epsilon$, there exists an algorithm that maintains 
$(1+\epsilon)$-approximate maximum weighted matching in
a graph where edges have weights between $[1, \maxweight]$
in $O(\sqrt{m} \maxweight \epsilon^{-2})$ update time.
\end{theorem}

\section{Approximate Weighted Matchings with Polylog Dependency on $\maxweight$}
\label{sec:weightedscaling}

We now show algorithms that approximate the maximum weighted
matching in time that depends on $\log{\maxweight}$ instead
of $\poly{\maxweight}$.
This reduced dependency on $\maxweight$ is a subject of study
in static algorithms since $\maxweight$ is often $\poly{n}$ or larger.

Our overall scheme is based on the data structure for weighted matchings
by Anand et al. \cite{AnandBGS12,AnandBGS-Arxiv12}.
Their algorithm maintains $\log\maxweight$ levels and
the edges are partitioned across various levels according to their weights.
A matching $\matchinglevel{l}$ is maintained at each level $l$, and
they gave a way to form a single matching $\combinedmatching$
from these $\log{\maxweight}$ matchings.
Algorithmically $\combinedmatching$ can be viewed as the
result of a greedy process which proceed in decreasing order
of levels and adds edges whenever possible.
Alternatively, it can be viewed as adding an edge $(u,v) \in \matchinglevel{l}$
to $\combinedmatching$ and removing all edges incident
to $u$ and $v$ from all $M_{l'}$s where $l' < l$.
At any update step, the matching maintained is equivalent to the
one generated in Figure~\ref{fig:combine}.

\begin{figure}[ht]
\qquad

\begin{procedure}[H]
$\combinedmatching = \emptyset$\;
Let $l_{\max}$ and $l_{\min}$ be the maximum and minimum level number respectively\;
\For{$l = l_{\max}$ to $l_{\min}$} {
	$ \combinedmatching=  \combinedmatching \cup M_l$\;
	\For{$(u, v) \in M_l$} {
		Remove all the edges adjacent to $u$ and $v$ from $M_{l'}$ such that $l' < l$
		
	}
}
\end{procedure}

\caption{Generating $\combinedmatching$}
\label{fig:combine}
\end{figure}


Anand et al. \cite{AnandBGS12,AnandBGS-Arxiv12} showed  
that the combined matching $\combinedmatching$ can be maintained on a dynamic graph
if the matching at each level $l$ can be maintained.
We will use their result as a black-box via. the following Lemma.
\begin{lemma}
\label{lem:lemmacmobine}
(\cite{AnandBGS-Arxiv12})
If the matching on each level
is maintained in $O(f(n))$ update time, then the overall matching can be maintained
in $O( f(n) \log \maxweight)$ update time.
\end{lemma}

In their work, $f(n) = O(\log n)$ due to the use of the dynamic maximal
matching data structure by Baswana et al.\cite{baswana2011fully},
which leads to a total bound of $O( \log n \log \maxweight)$.
We will substitute our algorithms in place of this algorithm, and investigate
different leveling schemes which lead to improved approximation ratios.
This comes at a cost of a higher value of $f(n) = O(\sqrt{m} \poly{\epsilon^{-1}})$,
which leads to a time of $O(\sqrt m \log{\maxweight} \poly{\epsilon^{-1}})$ per update.

In Section \ref{subsec:three}, we present a deterministic algorithm which
maintains a $(3+\epsilon)$-MWM in $O( \sqrt{m} \log N \epsilon^{-3})$ time,
and in Section \ref{subsec:arbitrary}, we given an alternate approach
which maintains a $(1+\epsilon)$-approximate MWM in
 $O(\sqrt{m} \log \maxweight \epsilon^{-2 - O(\epsilon^{-1})} )$
time per update. Note that in both the above algorithm, we will maintain
approximate MCM or MWM matching  at each level.
For this we can use the amortized and worst-case versions of our data structures described
in Sections~\ref{sec:algo}~and Appendix~\ref{sec:improvementsdetails} 
leading to corresponding types of final bounds for the above algorithm.

In many of our proofs, we will incur $(1 + O(\epsilon))$ multiplicative
error in several places.
As a result, the final approximation factors in our calculations will
often be $1 + c \epsilon$ for some constant $c$.
Such bounds can be converted to $1 + \epsilon$ approximations
by initiating the calls with smaller values of $\epsilon$.
As a result, we will omit these steps to simplify presentation.

\subsection{$(3 + \epsilon)$-Approximation Using Approx MCMs}
\label{subsec:three}

We first show that our data structure for maintaining $(1+\epsilon)$-MCMs
given in Theorem~\ref{thm:maxcardinality} can be used on each level.
The transformation for turning a MWM problem into a set of
$O(\log{\maxweight})$ MCM instances is based on a rounding
scheme by Eppstein et al. \cite{EppsteinLMS10,EppsteinLMS12}.
For a fixed value of $r$, we assign an edge $e$ with
$w(e) \in [ \alpha^{l+r}, \alpha^{l+r+1})$ to level $l$ where $\alpha$
is a constant which we will calculate later.
Note that the level of some edges can be $-1$, but our proof below
can extend to any negative level as well.
We define the rounded weight of an edge $e$ assigned to level
$l$ using:
\begin{align*}
w_r(e) \defeq \alpha^{l+r}
\end{align*}

Our analysis of the quality of $\combinedmatching$ is based on 
mapping each edge in $\combinedmatching$ to a set of edges in $M_{l}$'s.
For $e(u, v) \in \combinedmatching$ from level $l$, we define $\mathcal{R}(e)$ as:
\begin{align*}
\mathcal{R}(e) = \{e\} \cup \{ (x,y)~|~(x,y) \in M_{l'}
	\text{ where }l' < l\text{, and  }\{x, y\} \cap \{u, v\} \neq \emptyset \}
\end{align*}
In other words, $\mathcal{R}(e)$ contains edge $e$ and all those edges
adjacent to $u$ and $v$ from lower levels that were removed when
$(u,v)$ was added to $\combinedmatching$.
Note that $e$ is the only edge in $\mathcal{R}(e)$ from level $i$. 
And for all $l' < l$, there can be at most 2 edges from level $l'$ in $\mathcal{R}(e)$.
To simplify our notations, we will use $w(S)$ to denote
the total weight of a set of edges $S$
(that could be either $\combinedmatching$, $\optmatching$
or $\matchinglevel{l}$ for some $l'$)

For an edge $e \in \combinedmatching$, let $\Phi(e)$ denote
the total rounded weights of edges in $\mathcal{R}(e)$, i.e.,
$\Phi(e) = w_r(\mathcal{R}(e))$.
We can show that $\Phi(e)$ is closely related to  $w_r(e)$.

\begin{lemma}
\label{lem:charge}
For $e \in \combinedmatching$,
\begin{align*}
\Phi(e) \leq \frac{\alpha + 1}{\alpha - 1} w_r(e)
\end{align*}
\end{lemma}

\begin{proof}
Let $e \in \combinedmatching$ be on level $i$.
Since there are at most $1$ edge on level $i$ assigned to $e$
($e$ itself) and 
$2$ edges per level assigned to $e$ for each level $j < i$, we have:
\begin{align*}
\Phi(e) =&  \sum_{e' \in \mathcal{R}(e)} w_r(e')\\
      = &  w_r(e) + \sum_{j < i} \sum_{e'\in \matchinglevel{j} \& e' \in \mathcal{R}(e)} w_r(e')\\
      \leq & \alpha^{i + r} + \sum_{j<i} 2 \alpha^{j+r}\\
      \leq & \alpha^{i + r} \left(1 + 2 \sum_{j<i} \alpha^{j - i} \right)\\
      = & w_r(e) \left( 1 + 2 \frac{1}{\alpha - 1} \right)\\
      = & \frac{\alpha + 1}{\alpha - 1} w_r(e)
\end{align*}

\end{proof}

This allows us to relate the weight of $\combinedmatching$
to the weight of the optimum matching, $\optmatching$.

\begin{lemma}
\label{lem:weightopt}
\begin{align*}
(1 + \epsilon)\frac{\alpha + 1}{\alpha - 1} w(\combinedmatching)
\geq w_r(\optmatching)
\end{align*}
\end{lemma}

\begin{proof}

Let $\optmatchinglevel{i}$ denote the edges of $\optmatching$ at level $i$.
Since $\matchinglevel{i}$ is a $(1+\epsilon)$ approximate matching at level $i$,
we have:
\begin{align*}
|\optmatchinglevel{i}| \leq& (1+\epsilon) |\matchinglevel{i}|\\
w_r(\optmatchinglevel{i}) \leq& (1+\epsilon) w_r(\matchinglevel{i})
	\qquad \text{Since edges on same level have the same values of $w_r(e)$}\\
w_r(\optmatching) \leq & (1 + \epsilon) \sum_{i} w_r(\matchinglevel{i})
\end{align*}

Consider an edge $e = (u,v) \in M_i$.
If $e \in \combinedmatching$, then $e \in \mathcal{R}(e)$. If $e \notin \combinedmatching$, then 
there exists an edge $e' \in \combinedmatching$ at level $j > i$ such that one of the endpoints of $e'$ is either $u$ or $v$, which means $e$ is in the set $\mathcal{R}(e')$.
Therefore each edge $e$ can be mapped to one or more $\mathcal{R}(e')$,
and we have:
\begin{align*}
\Phi(\combinedmatching) \geq \sum_{i} w_r(\matchinglevel{i})
\end{align*}
Which implies $(1 + \epsilon) \Phi(\combinedmatching) \geq w_r(\optmatching)$.
Summing Lemma \ref{lem:weightopt} over all edges in $\combinedmatching$
then gives:
\begin{align*}
(1 + \epsilon) \frac{\alpha + 1}{\alpha - 1}
w_r(\combinedmatching) \geq w_r(\optmatching)
\end{align*}
And the result follows from the fact that the rounded down
edge weights satisfy $w_r(e) \leq w(e)$.
\end{proof}

Hence, it suffices to find ratio between
$w_r(\optmatching)$ and $w(\optmatching)$.
The analysis in Anand et al.\cite{AnandBGS-Arxiv12} bounded this
ratio over a uniformly random choices of $r$.
They showed that the expected rounded value of the optimum matching,
$\textbf{E}_{r}[w_r(\optmatching)]$ satisfies
$\displaystyle\textbf{E}_{r}[w_r(\optmatching)]
\geq \frac{\alpha-1}{\alpha \ln \alpha} w(\optmatching)$,
which when combined with Lemma~\ref{lem:weightopt} leads
to an expected approximation ratio of about $3+\epsilon$
when $\alpha \approx 5.704$.
Here we show instead that a deterministic and worst-case bound
can be obtained by using $O(1/\epsilon)$ versions of our data structure,
each with a pre-selected value of $r$. 

We have $k= \ln{\alpha} / \ln( 1 + \epsilon)$ copies of our algorithm which work exactly identically but with different 
value of $r$.
For the $j$\textsuperscript{th} copy, $r(j) = \frac{j - 1}{k}$. Consider an edge $e$ such
that $w(e) = \alpha^{i + \delta}$ where $0 < \delta \le 1$. Let $j^*$ is the value such that
$\frac{j^* - 1}{k} \leq \delta < \frac{j^*}{k}$. Then we have:
\begin{align*}
w_{r(j)}(e) =&
\left\{
\begin{array}{lr}
\alpha^{i + \frac{j-1}{k}} & \text{if $j \leq j^*$} \\
\alpha^{i + \frac{j-1}{k} - 1} & \text{$j > j^*$}
\end{array}
\right.
\end{align*}
Informally, an edge $e$ is at level $i$ in j\textsuperscript{th} copy, if 
$j \leq j^*$ otherwise it is at level $i-1$. We want to relate the
weight of maximum matching $\optmatching$ in $G$ to the new weight 
in these $k$ copies. Specifically, we want to get a relation similar to the
relation between $\displaystyle\textbf{E}_{r}[w_r(\optmatching)]$ and
$w(\optmatching)$ mentioned above.  We show that there 
exists a $j$ with the following relation.
\begin{lemma}
\label{lem:deterministic}
There  exists a $j$ such that:
\begin{align*}
w_{r(j)}(\optmatching)
\geq (1-\epsilon)  \frac{\alpha-1}{\alpha \ln \alpha} w(\optmatching)
\end{align*}
\end{lemma}

\begin{proof}
Summing over all $j$ of $w_{r(j)}(e)$ gives:
\begin{align*}
\frac{\sum_{j = 1}^k  w_{r(j)}(e)}{w(e)}
= & \frac{\sum_{j = 1}^{j^*} \alpha^{i + \frac{j-1}{k}}
	+ \sum_{j = j^{*} + 1}^{k} \alpha^{i + \frac{j-1}{k} - 1}}
	{\alpha^{i + \delta}} \nonumber\\
= & \frac{\sum_{j = 1}^{j^*} \alpha^{\frac{j-1}{k}}
	+ \sum_{j = j^{*} + 1}^{k} \alpha^{\frac{j-1}{k} - 1}}
	{\alpha^{\delta}} \nonumber \\
	= & \alpha^{-\delta + \frac{j^*-1}{k}} \left( \sum_{j = 1}^{j^*} \alpha^{j/k - j^{*}/k}
	+ \sum_{j = j^{*} + 1}^{k} \alpha^{j/k - j^{*}/k- 1} \right) \nonumber \\
= & (1+\epsilon)^{-k \delta + j^{*}-1} \left( \sum_{j = 1}^{j^*} (1 + \epsilon)^{j - j^{*}}
	+ \sum_{j = j^{*} + 1}^{k} (1 + \epsilon)^{j - j^{*}- k} \right)
\end{align*}
Since $j^*$ was chosen such that $\frac{j^*}{k} > \delta$,
$-k \delta + j^{*}-1 \geq -k (\frac{j^{*}}{k}) + j^*-1 = -1$
and
$(1+\epsilon)^{-k \delta + j^{*}-1} \geq (1+\epsilon)^{-1}$.
Substituting this gives:
\begin{align*}
\frac{\sum_{j = 1}^k  w_{r(j)}(e)}{w(e)} \geq &  (1+\epsilon)^{-1} \left( \sum_{j = 1}^{j^*} (1 + \epsilon)^{j - j^{*}}
	+ \sum_{j = j^{*} + 1}^{k} (1 + \epsilon)^{j - j^{*}- k} \right)
\end{align*}
The two summations is a rearranged version of a geometric sum.
It can be rearranged by substituting $l= j^*-j+1$ and $l=j^*-j+k+1$
in the first and second summation respectively to obtain:
\begin{align*}
\frac{\sum_{j = 1}^k  w_{r(j)}(e)}{w(e)} = &  (1+\epsilon)^{-1} \left(\sum_{l = 1}^{j^*} (1 + \epsilon)^{-l+1} +
       \sum_{l = j^* + 1}^{k} (1 + \epsilon)^{-l+1} \right)  \nonumber\\
= & \sum_{l = 1}^{k} (1 + \epsilon)^{-l} \nonumber\\
= &   \frac{1 - (1 + \epsilon)^{- k }}{\epsilon} \nonumber\\
= &   \frac{(1 - 1/\alpha)}{\epsilon} \nonumber \\
= &   \frac{\alpha - 1}{\alpha \epsilon}
\end{align*}

Summing this over all edges in $\optmatching$ gives:
\begin{align*}
\sum_{e \in \optmatching} \sum_{j} w_{r(j)}(e) \geq & \sum_{e \in \optmatching} \frac{\alpha - 1}{\alpha \epsilon} w(e) \nonumber\\
\sum_{j} w_{r(j)}(\optmatching) \geq & \frac{\alpha - 1}{\alpha \epsilon} w(\optmatching) \nonumber\\
\end{align*}

By an averaging argument we get:
\begin{align*}
\max_{j} \{ w_{r(j)}(\optmatching) \}
\geq & \frac{1}{k} \sum_{j} w_{r(j)}\optmatching \nonumber\\
\geq &  \frac{\alpha - 1}{\alpha \epsilon k} w(\optmatching)
\end{align*}
Note that $k = \ln{\alpha} / \ln( 1 + \epsilon)$.
Here we make use of the following known fact about the behavior
of the log function around $1$:

\begin{fact}
\label{fact:log}
For $\epsilon < 1$, if $0 \leq x \leq \epsilon$, then $\ln(1+x) \geq (1-\epsilon)x$.
\end{fact}
Applying it with $x = \epsilon$ gives:
\begin{align*}
\max_{j} \{ w_{r(j)}(\optmatching) \}
= &  \frac{(\alpha - 1)\ln(1 + \epsilon)}
	{\alpha \epsilon \ln{\alpha}} w(\optmatching)\\
\geq &   \frac{(\alpha - 1) (1 - \epsilon) \epsilon}
	{\alpha \epsilon \ln{\alpha}} w(\optmatching)
		\qquad \text{By Fact~\ref{fact:log}}\\
= & (1 - \epsilon)\frac{\alpha - 1}
	{\alpha \ln{\alpha}} w(\optmatching)
\end{align*}


\end{proof}

Combining Lemmas~\ref{lem:weightopt}~and~\ref{lem:deterministic}
gives the following theorem.

\begin{theorem}
\label{thm:threewmw}
For any $\epsilon < 1/2$, there exists a fully dynamic algorithm that maintains a
$(3 + \epsilon)$-MWM  for any graph on $n$ in worst
case $O(\sqrt{m} \log \maxweight \epsilon^{-3})$ time per update.
\end{theorem}

\begin{proof}
Consider maintaining $k$ copies of our data structure 
and picking the maximum
weighted matching among these copies as the current
best matching.

Using Lemma~\ref{lem:weightopt}, we get:
\begin{align*}
\forall j \qquad (1 + \epsilon)\frac{\alpha + 1}{\alpha - 1}
	w(\combinedmatching(j)) \geq & w_{r(j)}(\optmatching) \\
\end{align*}


Using Lemma~\ref{lem:deterministic}, there exists a $j'=\arg\max_{j} \{ w(\combinedmatching(j))\}$ such that 
\begin{align*}
w_{r(j')}(\optmatching)
\geq (1-\epsilon) \frac{\alpha-1}{\alpha \ln \alpha} w(\optmatching)
\end{align*}
Combining the above two equations we get:
\begin{align*}
 (1 + \epsilon)\frac{\alpha + 1}{\alpha - 1}
	w(\combinedmatching(j')) \geq & (1-\epsilon) \frac{\alpha-1}{\alpha \ln \alpha} w(\optmatching) \\
\left(\frac{1+\epsilon}{1-\epsilon}\right) \frac{(\alpha + 1) \alpha \ln\alpha}{(\alpha - 1)^2}
	 w(\combinedmatching(j')) \geq & w(\optmatching)
\end{align*}

Where one can check that $\frac{1+\epsilon}{1-\epsilon} \leq (1 + 4\epsilon)$ when $\epsilon < 1/2$.
By a suitable choice of $\epsilon$, this factor of $1 + 4 \epsilon$ can be turned
into $1 + \epsilon'$.
This implies that the approximation ratio obtained by our algorithm is 
$\frac{(1+\epsilon) (\alpha + 1) \alpha \ln \alpha}{(\alpha-1)^2}$.
This term achieves its minimum value of $\approx 3+3\epsilon$ when
$\alpha \approx 5.704$. Again this approximation ratio can be turned into $3+\epsilon'$
by a suitable choice of $\epsilon$.

For the update time, note that since $\alpha$ is a constant, $k = O( 1/ \log( 1+ \epsilon))= O(1 / \epsilon)$ copies of
the structure are needed. In each such copy, a matching can be maintained in $O( \sqrt{m} \log \maxweight \epsilon^{-2})$
update time. So matching in all the copies can be maintained in $O( \sqrt{m} \log \maxweight \epsilon^{-3})$ time per update.
\end{proof}

\ignore{The relevant lemma based on works by Eppstein et al.
\cite{EppsteinLMS12} is:
\begin{lemma}
\label{lem:randomrounding}
For an edge $e$, $E_r[w_r(e)/ w(e)] = \frac{\alpha-1}{\alpha \log \alpha}$.
Also $w(e) \ge w_r(e)$.
\end{lemma}

\begin{proof}
Let $w(e) = \alpha^{i+\delta}$ where $i$ is an integer and $0 < \delta \le 1$. 

So $   w_r(e) = \begin{cases}
		\alpha^{r+i}, & \text{if} \ \ r \le \delta ,\\
        \alpha^{r+i-1}, & \text{if} \ \ r > \delta 
        \end{cases}
 $ 
 
The expected value can be calculated as:\\
\begin{tabular}{lllll}
$E_r[w_r(e)/w(e)]$ & = & $\displaystyle\int_{0}^{\delta} \frac{\alpha^{r+i}}{\alpha^{i+\delta}} dr + \int^{1}_{\delta} \frac{\alpha^{r+i-1}}{\alpha^{i+\delta}} dr$\\
 			& = & $\displaystyle\frac{1}{\alpha^{\delta} \ln \alpha} \Big(\big(\alpha^{r}\big)_{0}^{\delta} + \big( \alpha^{r-1} \big)_{\delta}^{1}\Big)$\\
 			& = & $\displaystyle\frac{1}{\alpha^{\delta} \ln \alpha} \Big(\big(\alpha^{\delta}-1\big) + \big( 1-\alpha^{\delta-1} \big)\Big) $\\\\
 			& = & $\displaystyle\frac{\alpha-1}{\alpha \ln \alpha}$
\end{tabular}

For the second statement of the lemma, if $e$ is at level $i$
then $w(e) \in [ \alpha^{i+r}, \alpha^{i+r+1})$ and $w_r(e) = \alpha^{i+r}$.
So $w(e) \ge w_r(e)$.
\end{proof}

Applying Lemma \ref{lem:randomrounding} over all edges of $\optmatching$
and summing them using linearity of expectation gives
$E_{r}[w_r(\optmatching)]
\geq \frac{\alpha-1}{\alpha \ln \alpha} w(\optmatching)$.
}

\subsection{$(1 + \epsilon)$-MWMs Using Approximate MWMs}
\label{subsec:arbitrary}
\subsubsection*{Overview}
In this section, we present an algorithm that maintains a
$(1 + \epsilon)$-MWM using a more gradual bucketing scheme.
We start by observing the definition of $\mathcal{R}(e)$ for an 
edge $e(u,v)$ in $\combinedmatching$ from the previous section.
Informally, $\mathcal{R}(e)$ contains edge $e$ and all those edges
adjacent to $u$ and $v$ from lower levels that were removed when
$(u,v)$ was added to $\combinedmatching$. A closer look at 
our algorithm reveals that the approximation ratio depends on the
ratio of weight of $e$ and the combined  weight of 
edges in  $\mathcal{R}(e)$. This ratio can be reduced if the edges at
lower level have significantly less weight than the weight of edge $e$.
To achieve this, we will artificially create levels such that 
the ratio of weight between two consecutive level is significant. 
For this, we will drop some edges from the graph to create 
a {\em gap} between two consecutive levels.
In order to account for the weight of these dropped edges, we in turn need to
keep several copies of our data structure with different edges left out 
in the other copies.

We then proceed with the same algorithm as mentioned in Section~\ref{sec:weightedscaling}
with one main difference. Instead of maintaining $(1+\epsilon)$-MCM at each 
level, we maintain $(1+\epsilon)$-MWM at each level using the Theorem~\ref{thm:weightedsimple}.
Note that this theorem has a dependence of $\maxweight$ in its running time. We will show that 
each level can be formed in such a way that $\maxweight$ can be bounded by $O( \epsilon^{-O(\epsilon^{-1})} )$.
So the running time for maintaining $(1+\epsilon)$ approximate MWM at each level 
will have exponential dependence
on $(1/\epsilon)$.

Thereafter, we combine the matching across the various level using the same procedure 
as mentioned in  Section~\ref{sec:weightedscaling}. We will show that there exists a 
copy of our data structure such that the weight of the matching maintained by our
algorithm in that copy is a good approximation of maximum weighted matching 
in the entire graph.
\subsubsection*{Algorithm}
Once again we partition  the edges
by weights geometrically:
an edge $e$ is in bucket $b$ if $w(e)$ is in the range
$[ \epsilon^{-b}, \epsilon^{-(b+1)} )$.
However, our levels no longer corresponds to individual buckets, but instead
to a set of $C - 1$ continuous buckets for value of $C$ to be specified.
We will also remove some of these buckets, and the choices of buckets
to remove leads us to run several copies of our data structure
simultaneously.

We will run $C = \lceil \epsilon^{-1} \rceil$ copies of our algorithm,
where in the $c$\textsuperscript{th} copy, we remove all buckets $i$ such that
$i \mod C= c$.
This leads to a set of graphs $G^{0} \ldots G^{C - 1}$.
Removing the buckets creates natural partitions of the
remaining edges, which gives our levels.
For a copy $c$, we will place buckets with
$b = [ lC + c + 1 \ldots (l+1) C + c -1]$ into level $l$.
Note that the ratio of maximum to minimum edge weight
in each level is bounded by $\epsilon^{- (C - 1)}( = O( \epsilon^{-O(\epsilon^{-1}} )$).
Therefore, the algorithm given in Theorem \ref{thm:weightedsimple}
allows us to maintain an $(1+\epsilon)$-MWM in 
$O(  \epsilon^{-O(\epsilon^{-1})} \sqrt{m} \epsilon^{-2} \log(\epsilon^{-1}))
=O(  \epsilon^{-2-O(\epsilon^{-1})} \sqrt{m} \log(\epsilon^{-1}))$ time at each level.
These matchings can in turn be combined together in the same
way as in Section \ref{sec:weightedscaling}.
An illustration of levelling scheme used by our algorithm is shown in
Figure~\ref{fig:bucketlevelexample}.

 \begin{figure}
 \begin{tikzpicture}
\tikzstyle matched=[red,very thick]
\tikzstyle vertical=[dashed]
\tikzstyle inH=[very thick,blue]
\tikzstyle inM=[very thick]
\draw (0, 0) -- (10, 0);
\draw (0, .5) -- (10, .5);
\draw (0, 1) -- (10, 1 );
\draw (0, 1.5) -- (10, 1.5);
\draw (0, 2) -- (10, 2);
\draw (0, 2.5) -- (10, 2.5);
\draw (0, 3) -- (10, 3);
\draw (0, 3.5) -- (10, 3.5);
\draw (0, 4) -- (10, 4);
\draw (0, 4.5) -- (10, 4.5);
\draw (0, 5) -- (10, 5);
\draw (0,5.5) -- (10, 5.5);

\draw[dotted,thick] (3,5.7) -- (3,7);
\draw[dotted,thick] (6,5.7) -- (6,7);

\draw (11, -0.5) node {Bucket};

\draw (11, 0.25) node {0};
\draw (11, 0.75) node {1};
\draw (11, 1.25) node {2};
\draw (11, 1.75) node {3};
\draw (11, 2.25) node {4};
\draw (11, 2.75) node {5};
\draw (11, 3.25) node {6};
\draw (11, 3.75) node {7};
\draw (11, 4.25) node {8};
\draw (11, 4.75) node {9};
\draw (11, 5.25) node {10};
\draw (11, 5.75) node {11};


\draw [fill=gray!40] (1,.5) rectangle (3,1.5) ;
\draw (2, 1) node {Level 0};
\draw [fill=gray!40] (1,2) rectangle (3,3) ;
\draw (2, 2.5) node {Level 1};
\draw [fill=gray!40] (1,3.5) rectangle (3,4.5) ;
\draw (2, 4) node {Level 2};
\draw [fill=gray!40] (1,5) rectangle (3,5.5) ;
\draw (1,5.5) -- (1,5.7);
\draw (3,5.5) -- (3,5.7);

\draw [fill=gray!20] (7,0) rectangle (9,1) ;
\draw (8, 0.5) node {Level 0};
\draw [fill=gray!20] (7,1.5) rectangle (9,2.5) ;
\draw (8, 2) node {Level 1};
\draw [fill=gray!20] (7,3) rectangle (9,4) ;
\draw (8, 3.5) node {Level 2};
\draw [fill=gray!20] (7,4.5) rectangle (9,5.5) ;
\draw (8, 5) node {Level 3};


\draw [fill=gray!60] (4,0) rectangle (6,.5) ;
\draw (5, .25) node {Level 0};
\draw [fill=gray!60] (4,1) rectangle (6,2) ;
\draw (5, 1.5) node {Level 1};
\draw [fill=gray!60] (4,2.5) rectangle (6,3.5) ;
\draw (5, 3) node {Level 2};
\draw [fill=gray!60] (4,4) rectangle (6,5) ;
\draw (5, 4.5) node {Level 3};

\draw (-.5,-.5) node {Copy};
\draw (2,-.5) node {$c=0$};
\draw (5,-.5) node {$c=1$};
\draw (8,-.5) node {$c=2$};

\end{tikzpicture}
\caption{Bucketing and level scheme where $C = 3$}
\label{fig:bucketlevelexample}
\end{figure}
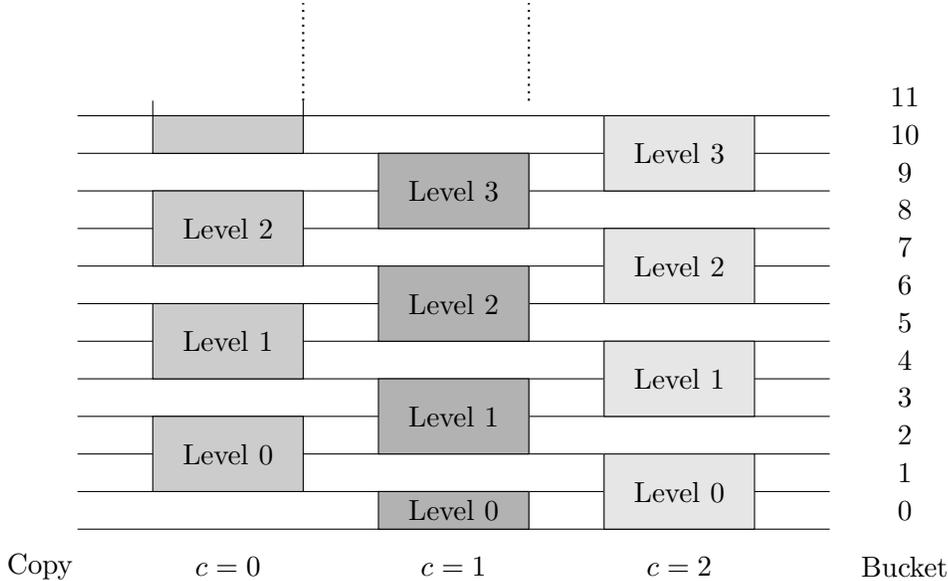

We start by analyzing the guarantees of our algorithm on the c\textsuperscript{th} copy.
Specifically, the approximation ratio of the combined matching
$\combinedmatchingcopy{c}$ 
w.r.t. the maximum matching $\optmatchingcopy{c}$ in this copy .
Let $\matchinglevelcopy{l}{c}$,
and  $\optmatchinglevelcopy{l}{c}$ be the edges of  $\combinedmatchingcopy{c}$
and $\optmatchingcopy{c}$ at level $l$ respectively.
Once again, for an edge $e=(u,v)$ in $\matchinglevelcopy{l}{c}$
we define $\mathcal{R}^{c}(e)$ as:
 \begin{align*}
\mathcal{R}^{c}(e) = \{e\} \cup \{ (x,y) | (x,y) \in \matchinglevelcopy{l'}{c}
	\text{ where }l' < l	\text{, and  }\{x, y\} \cap \{u, v\} \neq \emptyset \}
\end{align*}

For an edge $e \in \combinedmatchingcopy{c}$, let $\Phi^{c}(e)$ denote
the total rounded weights of edges in $\mathcal{R}^{c}(e)$, i.e.,
$\Phi^{c}(e) = w(\mathcal{R}^{c}(e))$.
We can show that $\Phi^{c}(e)$ is related to $w(e)$ by the following inequality:

\begin{lemma}
\label{lem:extendedcharge}
For any edge $e$ in the combined matching
$\combinedmatchingcopy{c}$, we have:
\begin{align*}
\Phi^{c}(e) \leq \left( 1 + 3\epsilon \right)w(e)
\end{align*}
\end{lemma}

\begin{proof}
Assume that $e$ is on level $l$.
Since there are at most $1$ edge on level $l$ assigned to $e$
($e$ itself) and $2$ edges per level $l' < l$ assigned to $e$, we have:
\begin{align*}
\Phi^{c}(e) =&  \sum_{e' \in \mathcal{R}^{c}(e)} w(e')\\
 	   = &  w(e) + \sum_{l' < l} \sum_{e'\in \matchinglevelcopy{l'}{c} \cap
		\mathcal{R}^{c}(e)} w(e')
\end{align*}
Since an edge on level $l'$ is in bucket $[l'C + c + 1 \ldots (l'+1) C + c -1]$
and an edge in bucket $b$ has weight at most $(1 / \epsilon)^{b + 1}$,
the weight of an edge at level $l'$ is $\le \epsilon^{-c-(l'+1)C}$.
\begin{align*}
\Phi^{c}(e)
      \leq & w(e)+ \sum_{l'<l} 2 \epsilon^{-c - (l' + 1)C}
\end{align*}
Which can in turn be bounded relative to $w(e)$. Since $e$ is in level $l$
and an edge in bucket $b$ has weight at least $(1 / \epsilon)^{b}$, 
$w(e) \ge \epsilon^{-lC-c-1}$
\begin{align*}
\Phi^{c}(e)
      \leq & w(e)+ \sum_{l'<l} 2 \epsilon^{(l - l' - 1)C + 1} w(e)\\
      \leq & w(e)+ \frac{2 \epsilon}{1 - \epsilon^C} w(e)\\
      \leq & w(e) (1 + 3\epsilon)\qquad\text{assuming that $\epsilon < 1/2$}
\end{align*}
\end{proof}


We now show the relation between the combined matching 
$\combinedmatchingcopy{c}$ and $\optmatchingcopy{c}$.

\begin{lemma}
\label{lem:extendedweightopt}
\begin{align*}
(1 + 7 \epsilon ) w(\combinedmatchingcopy{c}) \geq w(\optmatchingcopy{c})
\end{align*}
\end{lemma}

\begin{proof}

Since $w(\matchinglevelcopy{l}{c})$ is a $(1+\epsilon)$ approximate
maximum weighted matching on level $l$, we have:
\begin{align*}
w(\optmatchinglevelcopy{l}{c})
\leq & (1+\epsilon) w(\matchinglevelcopy{l}{c})\\
w(\optmatchingcopy{c})
\leq & (1 + \epsilon) \sum_{l} w(\matchinglevelcopy{l}{c})
\end{align*}

Consider an edge $e = (u,v) \in \matchinglevelcopy{l}{c}$.
If $e \in \combinedmatchingcopy{c}$, then $e \in \mathcal{R}^{c}(e)$.
If $e \notin \combinedmatchingcopy{c}$, then there exists an edge $e'$
at level $l' > l$ such that one of the endpoints of $e'$ is either $u$ or $v$,
which means $e$ is in the set $\mathcal{R}^{c}(e')$.
Therefore each edge $e$ can be mapped to one more
$\mathcal{R}^{c}(e')$, and we have:
\begin{align*}
\Phi^{c}(\combinedmatchingcopy{c}) \geq \sum_{l} w(\matchinglevelcopy{l}{c})
\end{align*}
Which implies $(1 + \epsilon) \Phi^{c}(\combinedmatchingcopy{c})
\geq w(\optmatchingcopy{c})$.
Summing Lemma~\ref{lem:extendedcharge} over all edges in
$\combinedmatchingcopy{c}$ then gives:
\begin{align*}
(1 + \epsilon)(1 + 3\epsilon) w(\combinedmatchingcopy{c})
\geq w(\optmatchingcopy{c})
\end{align*}
And the bound follows from $(1 + \epsilon)(1 + 3\epsilon)
= 1 + 4 \epsilon + 3 \epsilon^2 \leq 1 + 7 \epsilon$
when $\epsilon < 1$.
\end{proof}

We now find a relation between $w(\optmatchingcopy{c})$ and
$w(\optmatching)$.
We show there is at least one copy whose maximum matching
has weight at least $(1 - 1/C)$ of $w(\optmatching)$.

\begin{lemma}
\label{lem:goodcopy}
At any update step, if the maximum weight matching in the current graph is
$w(\optmatching)$, there exist a copy $c$ such that
$w(\optmatchingcopy{c}) \geq (1 - 1/C) w(\optmatching)$.
\end{lemma}

\begin{proof}

Let $\bar{\optmatching}^{c}$ denote the set of edges in $\optmatching$
that are not present in the $c$\textsuperscript{th} copy.
Since each bucket is removed in only one copy, we have:
\begin{align*}
\cup_{c} \bar{\optmatching}^{c} =& \optmatching\\
\sum_{c} w(\bar{\optmatching}^{c}) =& w(\optmatching)
\end{align*}
Since $\optmatching \setminus \bar{\optmatching}^{c}$ is a
matching in $G^{c}$, we have
$w(\optmatchingcopy{c}) \geq w(\optmatching) - w(\bar{\optmatching}^{c})$.
Note that the inequality is due to $\optmatchingcopy{c}$ being the maximum
weighted matching in $G^{c}$ instead of the restriction of $\optmatching$ on it.
Summing over all $c$ copies gives:
\begin{align*}
\sum_{c} w(\optmatchingcopy{c}) 
\geq & \sum_{c} \left( w(\optmatching) - w(\bar{\optmatching}^{c}) \right)\\
= & C \cdot w(\optmatching) - \left(\sum_{c} w(\bar{\optmatching}^{c}) \right) \\
= & (C - 1) \cdot w(\optmatching)
\end{align*}
Dividing both sides by $C$ gives that the average of $w(\optmatchingcopy{c})$
is at least $(1 - 1/C) w(\optmatching)$.
Therefore there exist some $c$ where
$w(\optmatchingcopy{c}) \geq (1 - 1/C) w(\optmatching)$.
 
\end{proof}

Combining Lemmas~\ref{lem:goodcopy}~and~\ref{lem:extendedweightopt}, we deduce 
that there exists a copy $c$ such that\\ 
$(1 + 7 \epsilon ) w(\combinedmatchingcopy{c}) \geq (1-\epsilon) w(\optmatching)$. So,\\
\begin{align*}
 \frac{w(\optmatching)}{ w(\combinedmatchingcopy{c})} \le \frac{(1 + 7 \epsilon )}{(1-\epsilon)} 
\end{align*}
This ratio is less that $(1+16\epsilon)$ for $\epsilon < 1/2$.
By a suitable chose of $\epsilon'$, the factor of $(1+16\epsilon)$ can be 
turned into $(1+\epsilon')$.

This means that if we set $C = \lceil \epsilon^{-1} \rceil$ and maintain
$(1 + \epsilon)$-MWMs on each copy of our data structure,
then one of the maximum weight matching among these $C$ copies will always be
a good approximation of the maximum weighted matching for the entire graph.

Note that in each copy of our 
data structure there are
$O(\log_{\epsilon^{-1}} \maxweight) / C =
O\big(\frac{\log \maxweight}{C\log(\epsilon^{-1})}\big) $ levels
and in each level an approximate MWM is maintained in
$O(  \epsilon^{-2-O(\epsilon)^{-1}} \sqrt{m} \log(\epsilon^{-1}) )$ time.
This implies that the overall update time taken by our
algorithm across the $C$ copies is
$O(C \cdot \frac{\log \maxweight}{C \log(\epsilon^{-1})}
	\sqrt{m} \epsilon^{-2-O(\epsilon)^{-1}} \log(\epsilon^{-1}))
= O(\sqrt{m} \epsilon^{-2-O(\epsilon)^{-1}} \log{\maxweight} )$.
So we can state the following theorem:
\begin{theorem}
For any $\epsilon<1/2$, there exists a fully dynamic algorithm that maintains a $(1 + \epsilon)$-MWM
in worst case $O(  \sqrt{m} \epsilon^{- 2 - O(1/\epsilon)} \log{\maxweight  } )$
time per update. 
\end{theorem}




\section{Conclusion}
\label{sec:conclusion}

We showed a simpler method for maintaining approximate
matchings that maintains $(1 + \epsilon)$-approximations
in about $\sqrt{m}$ time per update.
Natural directions for future works are whether the update
time can be improved, and whether the exponential dependency
on $\epsilon^{-1}$ in the weighted case can be reduced.
Also, it would be interesting to explore whether this type of
approach can also maintain approximations of other
optimization objectives.

Theoretically our arbitrary quality approximation algorithm from
Section~\ref{subsec:arbitrary} outperforms the $(3 + \epsilon)$
approximation given in Section~\ref{subsec:three}.
However, its exponential dependency on $\epsilon^{-1}$ makes it
fall short of a practical algorithm for maintaining $(1+\epsilon)$-MWMs.
We believe that a more intricate rounding scheme such as the one
given in Section~\ref{subsec:arbitrary}, or possibly a data structure
that incorporates details of the Duan et al. algorithm
\cite{DuanP10, DuanPS11} are promising approaches in this direction.

\section*{Acknowledgements}
The authors wish to thank Surender Baswana, Gary Miller, Krzysztof Onak,
Sandeep Sen, and Danny Sleator for their very helpful comments
and discussions.
Richard Peng is supported by a Microsoft Research PhD Fellowship. Manoj Gupta
is supported by Microsoft Research India PhD Fellowship.

\bibliographystyle{alpha}
\bibliography{references}

\pagebreak

\begin{appendix}

\section{Improvements, Worst-Case Bound, and Weights }
\label{sec:improvementsdetails}

We now give the full details for the improvements
described in Section~\ref{subsec:improvements}.
This includes removing the dependency on other data
structures for finding $\vcover$ in Section~\ref{subsec:vcover},
incorporating weights in Section~\ref{subsec:weights},
and obtaining worst case guarantees in Section~\ref{subsec:worst}.
Each of these improvement can be made independent
of the other ones.
However, we will only give the pseudocode for a version
incorporating all three extensions in Section~\ref{subsec:worst}.

\subsection{Simpler Maintenance of Small Vertex Cover}
\label{subsec:vcover}

We now present a simple algorithm to maintain a constant 
approximation of vertex cover to remove our dependence on 
another algorithm in Section~\ref{sec:algo}. 
Like our algorithm for maintaining approximate matching,
a natural approach for maintaining $\vcover$ is to use a static algorithm,
\textsc{ApproxCover} for it and then do nothing for certain number 
of update step till the current vertex cover is a good 
approximation of minimum vertex cover.
The \textsc{ApproxCover}  routine finds a vertex cover by finding a maximal
matching and reporting all the endpoints of matched edges as
the vertex cover of $G$.
This gives a cover whose size is at most twice of optimum,
and we can also bound its guarantees after a number of updates.
We will use $\vcover(i)$ to denote a vertex cover generated by
calling $\textsc{ApproxCover}(G(i))$. Let $M(i)$ be the maximal matching 
which is used to find the vertex cover $\vcover(i)$.
In order to use this cover at some subsequent update $j$,
we insert all vertices involved in edge insertions from 
update step $i$ to $j$ 
and denote this new cover as $\vcover(i \rightarrow j)$.

\begin{lemma}
\label{lem:coverstable}
If $\vcover(i)$ is a $2$-approximate minimum vertex cover on $G(i)$
and $j \leq i + 1/4M(i)$,
then $\vcover(i \rightarrow j)$ is an $5$-approximation to
the minimum vertex cover in $G(j)$.
\end{lemma}

\begin{proof}
We use that fact that vertex cover after $t$ updates
can decrease by at most $t$.
The guarantee of \textsc{ApproxCover} implies that the size
of the minimum vertex cover in $G(i)$ is at least $|\vcover(i)|/2$.
Since $M(i) \leq |\vcover(i)|$, the size
the minimum cover in $G(j)$ is at least
\begin{align*}
\frac{1}{2} |\vcover(i)| - \frac{1}{4} |\vcover(i)|
= & \frac{1}{4} |\vcover(i)|
\end{align*}
Also, since $M(i) \leq |\vcover(i)|$,
\begin{align*}
|\vcover(i \rightarrow j)|
\leq |\vcover(i)| + 1/4 M(i)
\leq 5/4 |\vcover(i)|
\end{align*}
Therefore $|\vcover(i \rightarrow j)|$ is a $5$-approximation to
the minimum vertex cover in $G(j)$.
\end{proof}

To maintain a small vertex cover, we find 
a new vertex cover using $\vcover( i \rightarrow j)$.
Similar to Lemma~\ref{lem:smallcover}, we can build a {\em core subgraph} $G'$
by taking a small set of neighbors of each vertex in $\vcover( i \rightarrow j)$,
and prove the analogous result for maintaining vertex covers.

\begin{lemma}
\label{lem:coveronsmall}
Let $\vcover( i \rightarrow j)$ be a vertex cover in $G(j)$ generated 
as above from $\vcover(i)$, which was in turn returned by a
call to $\textsc{ApproxCover}(G(i))$.
Consider the core subgraph $G'$ having following edges.
\begin{itemize}
\item All edges between vertices in $\vcover( i \rightarrow j)$
\item Up to $|\vcover( i \rightarrow j)| + 1$ arbitrary edges to vertices in
$V \setminus \vcover( i \rightarrow j)$ for each vertex in $\vcover(i \rightarrow j)$.
\end{itemize}

$\textsc{ApproxCover}(G')$ returns a new vertex cover for $G(j)$,
$\vcover(j)$ whose size is at most twice the minimum.
\end{lemma}

Note that although $\vcover( i \rightarrow j)$ is a vertex cover
in $G(j)$, its may not be as good of an approximation due to
adding all vertices involved in insertions between updates $i$ and $j$.
Generating a new cover can be viewed as a way to reduce
the error.

\begin{proof}
Let $M'$ be a maximal matching found on $G'$
by \textsc{ApproxCover}.
It suffices to show that $M'$ is a maximal matching on $G(j)$
in a way similar to our proof of Lemma \ref{lem:smallcover}.
For contradiction, suppose there exists such a free edge
 $(u, v)$ in $G(j)$ with respect to matching $M'$ such that both $u$ and $v$ a
are unmatched.
As $\vcover(i \rightarrow j)$ is a vertex cover for $G(j)$, we may assume that
$u \in \vcover(i \rightarrow j)$ without loss of generality.
The maximality of $M'$ in $G'$ implies that $(u, v) \notin E(G')$,
and therefore $v  \notin V(G')$, which means $u$ has at least
 $|\vcover(i \rightarrow j)|+1$ neighbors in $V \setminus \vcover(i \rightarrow j)$.
As there are no edges between vertices in $V \setminus \vcover(i \rightarrow j)$,
each edge in $M'$ can only match one vertex from this set.
Also, $|M'| \leq |\vcover(i \rightarrow j)|$ since the size of any matching
is at most the size of a vertex cover.
This means $u$ has at least one unmatched neighbor in $G'$, 
contradicting the assumption that $M'$ is a maximal
matching in $G'$.
\end{proof}

Using Lemma~\ref{lem:coverstable} and \ref{lem:coveronsmall},
we see that the vertex cover in graph $G'$ can be found 
in $O( |\vcover( i \rightarrow j)|^2 )( = O( |M(j)|^2))$ time. This time can be amortized over the next 
$M(j)/4$ steps. Similar to the analysis in Section~\ref{sec:algo}, we can show 
that the amortized update time at each step is $O( \sqrt{m} )$ per update.

\ignore{
Let $\vcover^t$ be the vertex cover in $G_t$. To use the above algorithm
in the previous section, we need to construct $G'$ in $O((|\vcover^t|)^2)$.
Consider the
algorithm in the previous section. It finds a matching $M$ in $G$.
Then it does nothing for $t = \epsilon |M| /2$ iterations. The vertex cover $\vcover$ in G 
is size $\le 2|M|$. We first show that $|\vcover^t| = O( |M| )$. 

}


\subsection{Incorporating Weights}
\label{subsec:weights}

Till now we have described dynamic algorithms for maintaining 
approximate matching in unweighted graphs. Now we show that 
the same technique can be used to maintain maximum weighted matching.
In order to handle weighted matchings, we will substitute the
weighted matching algorithm given in Lemma~\ref{lem:staticmwm}
as our static algorithm.
We also need to prove the analogs of 
Lemmas~\ref{lem:stable}~and~\ref{lem:smallcover}
for the weighted setting.

\begin{lemma}
\label{lem:stableweighted}
If all edge weights are in the range $[1, \maxweight]$ and $M$ is
an $(1 + \epsilon)$-MWM in $G(i)$,
then for $j \le  i + \frac{\epsilon'}{\maxweight} |M|$,
$M(i \setminus j)$ is an $(1 + 2\epsilon + 2\epsilon')$-MWM in $G_{j}$
\end{lemma}

\begin{proof}
Similar to the proof of Lemma \ref{lem:stable}, except
each insertion can increase $w(\optmatching)$ by at most $\maxweight$,
and each deletion can decrease $w(\optmatching)$ by at most $\maxweight$.
\end{proof}

\begin{lemma}
\label{lem:extendedsmallcover}
Let $\vcover \subseteq V$ be a vertex cover, and $G'$ be the core subgraph consisting of:
\begin{itemize}
\item All edges between vertices in $\vcover$
\item The top $|\vcover| + 1$ maximum weighted edges to vertices
in $V \setminus \vcover$ for each vertex in $\vcover$.
\end{itemize}
If $M'$ is a $(1 + \epsilon)$-MWM in $G'$, then it's
also a $(1 + \epsilon)$-MWM in $G$.
\end{lemma}

\begin{proof}
Similar to Lemma \ref{lem:smallcover}, but we need
to compare the weights of edges $(u, x)$ and $(u, v)$
where $x \in N_{V \setminus \vcover}(u)$ and $v \notin \vcover$.
By our choice of the $|\vcover|$ edges incident to $u$ being
the ones of maximum weight, we have $w(u, x) \geq w(u, v)$.
\end{proof}

These two lemmas allows us to adapt the algorithm from
the Section~\ref{sec:algo} for weighted matchings, with an
additional factor of $O(\maxweight)$ in the running time.

To construct $G'$, for each $u \in \vcover$ we need to examine
edges incident to it in decreasing order of weights until we have
either exhausted the list, or found $|\vcover| + 1$
ones in $V \setminus \vcover$.
In either case, it suffices to extract the $2|\vcover|+1$ edges
of maximum weight incident to $u$.
\ignore{
This operation can be done by maintaining a balanced search
tree (see e.g. \cite{CLRS09}) at each vertex, which incurs an
extra cost of $O(|\vcover|\log{n})$,
This overhead be bounded by $O(\max\{ |\vcover|^2, \sqrt{m} \}$,
which allows us to ignore it in the overall runtime analysis.
}

\subsection{Worst Case Runtime}
\label{subsec:worst}

In this section, we will present an algorithm which obtains a better worst case bound
for maintaining approximate matchings.
The algorithm is similar in spirit to the one described in Section \ref{sec:algo}.
That is, we first find a approximate weighted matching in graph $G(i)$ and then do nothing for the 
certain update step $j$ till the matching we had  obtained gives a good approximation
of optimum. In order to find an approximate weighted matching in $G(j)$, we perform the following 
tasks:


\label{def:statics}
\begin{itemize}
    \item Finding a new vertex cover $\vcover(j)$  using $\vcover(i \rightarrow j)$ ( Lemma \ref{lem:coveronsmall})
	\item Constructing {\em core subgraph} $G'$ from $\vcover(j)$ (Lemma \ref{lem:extendedsmallcover})
	\item Running $\textsc{ApproxMWM}(G', \epsilon)$
\end{itemize}

As $G'$ may be as large as $G$, the number of edges in $G'$ in worst case  
can be as high as $\Omega(m)$, and the call to \textsc{ApproxMWM}
can only be bounded by $O(m  \epsilon^{-1} \log(\epsilon^{-1}))$.
However,  Lemma~\ref{lem:stableweighted} implies that once an
$(1 + \epsilon/4)$-MWM is found, it will remain an $(1 + \epsilon)$-MWM
after $\epsilon|M|/4\maxweight $ updates.
Note that the algorithm does nothing in these $\epsilon |M| / 4 \maxweight$ update steps. 
In order to obtain a worst case bound, we perform useful computations in 
these update steps. Specifically, instead of finishing the three task mentioned 
above at one update step, we gradually work on these three task for these 
$\epsilon |M| / 4 \maxweight$ steps.

We now describe our algorithm in Figure \ref{fig:lazyImp}. Our algorithm 
works in rounds. A round consists of three tasks that we mentioned above:

\begin{figure}[ht]
\begin{procedure}[H]
	\nl \If{Update is a deletion}{
		\nl If $(u, v) \in M$, remove it from $M$\;
		\nl Notify call to \textsc{ApproxMWM} in current round to remove $(u, v)$ \label{ln:approxmwm}\\
		\nl Notify current call to \textsc{ApproxCover} to add $u$.
	}
	\nl Perform $O(\maxweight \sqrt{m} \epsilon^{-2} \log(\epsilon^{-1}))$ steps in the current round\;
	\nl \If{Current round finishes}{
		\nl Replace $M$ by the results of the call to $\textsc{ApproxMWM}$\;
		\nl Initiate the next round consisting of:
\begin{enumerate}
\item Construct {\em core subgraph} $G'$ using the result from the previous call to $\textsc{ApproxCover}$
\item $\textsc{ApproxMWM}(G', 1 + \epsilon / 8)$
\item $\textsc{ApproxCover}(G')$
\end{enumerate}
	}
\caption{UpdateImp(u, v)}
\end{procedure}

\caption{Maintaining $(1 + \epsilon)$-MWMs with Worst Case Guarantees}
\label{fig:lazyImp}
\end{figure}

\begin{itemize}
	\item Constructing {\em core subgraph} $G'$ from a vertex cover\\
	We obtain this vertex cover from the result of $\textsc{ApproxCover}(G')$ which 
	was called in the previous round.
	\item Running $\textsc{ApproxMWM}(G', \epsilon)$\\
	The result of this procedure is used as the approximate maximum weighted matching 
	returned by the algorithm for the next round.
	\item Constructing a new vertex cover on $G'$ using \textsc{ApproxCover(G')}\\
	The vertex cover returned by this procedure will be used for the next round.
\end{itemize}

To obtain a worst case bound, we will run these procedure gradually, i.e, we will
do only $O(\maxweight \sqrt{m} \epsilon^{-2} \log(\epsilon^{-1}))$ amount
of work at each update step.
However, this approach has a drawback. For example,
suppose that while computing $\textsc{ApproxMWM}(G', \epsilon)$, 
we added an edge $e$ to the matching at some update step. 
But subsequently, if this edge is deleted from the graph, we don't have
any way of reflecting this change in the future
iterations of $\textsc{ApproxMWM}(G', \epsilon)$. So we just remove
these edges from the result of $\textsc{ApproxMWM}(G', \epsilon)$
(See line~\ref{ln:approxmwm}).
Similar is the case for the procedure $\textsc{ApproxCover}(G')$.
However we will show that we don't lose much in terms of
approximation due to this workaround.
 
We will show that the number of update steps in a round are 
large enough so that all these three tasks can be performed if
the work done at each update step is $O(\maxweight \sqrt{m} \epsilon^{-2} \log(\epsilon^{-1}))$.
On the other hand, we also show that the number of update steps 
is small enough so that the approximate weighted matching returned 
in the previous round has a good approximate ratio till the current 
round ends.
\ignore{This means that as long as a static round finishes reasonably
fast, its result can still be used for a sequence of updates.
That is, instead of finish the static pass in the current update
and using the matching returned for $O(\epsilon|M|)$ updates,
we can finish it over the course of $O(\epsilon/2|M|)$
and use the matching returned for the next $O(\epsilon/2|M|)$
updates.
Of course, this means that over the $O(\epsilon / 2 |M|)$ updates
during which we're waiting for this pass to finish, we will need
to use results from an earlier pass instead.
One other issue is that the results produced by  a
pass need to take into account the updates that took
place while it was running.
Therefore, it is helpful to view them as tracking the
modifications done by edge insertions/deletions.

Running these passes more gradually leads to an algorithm
with $\sqrt{m}$ type worst-case runtime per update.
However, it leads to the issue of having several passes
taking place simultaneously.
While this can be resolved, we instead resort to an alternate
solution: using the completion of one pass to initiate the next.
In other words, the algorithm is `self-timed' using the bound
on its worst-case running time.
We ensure that the algorithm always performs a number
of number of steps in the current pass, and initiate the
next one upon the current one's completion.
An algorithm with all the modifications described in this
section incorporated is shown in Figure \ref{fig:lazyImp}.
}

We find the the number of updates it takes for a for a
round to finish using a proof
analogous to that of Theorem \ref{thm:lazy}

\begin{lemma}
\label{lem:singlepass}
If a round initiated at update $i$ with a vertex cover $\vcover$
satisfying $|\vcover| \leq 10 \optmatchingweight{i}$,
and is run for $O(\maxweight \sqrt{m} \epsilon^{-2} \log(\epsilon^{-1}) )$
steps per subsequent update,
it will finish within $\frac{\epsilon \optmatchingweight{i}}{8\maxweight}$ updates.
\end{lemma}

\begin{proof}
The number of edges in $G'$ can be bounded by
$\min\{m, O(|\vcover|^2)\} = O(\min\{m, \optmatchingweight{i}^2 \})$.
This means that the total cost of the current round is 
$O(\min\{m, \optmatchingweight{i}^2 \} \epsilon^{-1}\log(\epsilon^{-1}) )$.
For this to finish in $\frac{\epsilon \optmatchingweight{i}}{8\maxweight} $  updates,
the amount of work that needs to be performed at each update is:
\begin{align*}
&O\left(\min \{ m, \optmatchingweight{i}^2 \} \epsilon^{-1} \log(\epsilon^{-1}) \right)
	/ \left(\frac{\epsilon \optmatchingweight{i}}{8\maxweight} \right) \\
= & O\left(\maxweight
	\min \{ \frac{m}{\optmatchingweight{i}}, \optmatchingweight{i} \}  
	 \epsilon^{-2} \log(\epsilon^{-1})\right)
\end{align*}
Considering cases of $\optmatchingweight{i} \geq \sqrt{m}$
and $\optmatchingweight{i} < \sqrt{m}$ leads to an
$O(\maxweight \sqrt{m} / \epsilon^2 \log(\epsilon^{-1}))$ worst case time per update step.
\end{proof}

The worst case bounds can now be obtained by analyzing the duration
for which $M(i \setminus j)$ is used as $M$ in the algorithm, and proving
inductively the conditions required for the vertex cover in Lemma~\ref{lem:singlepass}.\\\\


We finish this section by giving the proof of Theorem~\ref{thm:weightedsimple}
\begin{proof}
We first show inductively that at the start of a round initiated at
update $j$ we have a vertex cover $\vcover$ satisfying
$|\vcover| \leq 10 \optmatchingweight{j}$.
The base case follows from both being $0$, and for the inductive
case, consider a round initiated at time $j$.
Suppose the current value of $\vcover$ was produced by a round
initiated at time $i < j$. So the vertex cover $\vcover$ at update step
$i$ is $\vcover(i)$ and at update step $j$, it is $\vcover(i \rightarrow j)$.
Lemma~\ref{lem:coveronsmall} implies that $\vcover(i)$ is a
$2$-approximate vertex cover on $G(i)$.
Using the conditions in Lemma \ref{lem:singlepass}, we have
$j - i \leq \frac{\epsilon \optmatchingweight{i}}{8\maxweight}
\leq 1/4 \optmatchingsize{i}$.
Therefore by Lemmas~\ref{lem:coverstable}, $|\vcover(i \rightarrow j)|$
is a 5-approximation of minimum vertex cover in $G(j)$.
Since the size of minimum vertex cover in $G(j)$ is $\le 2 \optmatchingsize{j}$, this implies that
$|\vcover(i \rightarrow j)| \le 10 \optmatchingsize{j}$.
Thus the requirements of Lemma~\ref{lem:singlepass}
is satisfied for all rounds.

For the guarantees on $M$, it suffices to prove that the matching
produced by a round initiated at update $i$ remains an
$(1 + \epsilon)$-MWM until the completion of the next round.
Let the update by which this round finishes be $i'$.
Applying Lemma~\ref{lem:singlepass} to all rounds gives:
\begin{align*}
i' \leq i + \frac{\epsilon \optmatchingweight{i}}{8\maxweight}
\end{align*}
and that the round initiated at update $i'$ finishes 
in $\frac{\epsilon \optmatchingweight{i'}}{8\maxweight}$ updates.
Therefore it suffices to show that $M(i \setminus j)$ is a $(1 + \epsilon)$-MWM
for $G_j$ for all $j \leq i +
\frac{\epsilon}{8\maxweight} (\optmatchingweight{i} + \optmatchingweight{i'})$.

Since the call to \textsc{ApproxMWM} is ran with error $\epsilon / 4$,
we have $\optmatchingweight{i} \leq (1 + \epsilon / 4) w(M(i))$,
and:
\begin{align*}
\optmatchingweight{i'}
\leq & \optmatchingweight{i} + \maxweight(i' - i) \nonumber\\
\leq & (1 + \epsilon / 8) \optmatchingweight{i} \nonumber\\
\leq & (1 + \epsilon / 8) (1 + \epsilon / 4) w(M(i))
\end{align*}
Summing these then gives:
\begin{align*}
\frac{\epsilon}{8\maxweight} (\optmatchingweight{i} + \optmatchingweight{i'})
\leq & \frac{\epsilon}{8\maxweight} ((1 + \epsilon / 4) w(M(i)) + (1 + \epsilon / 8) (1 + \epsilon / 4) w(M(i))) \nonumber\\
= & \frac{\epsilon}{8\maxweight} (2 + \epsilon / 8) (1 + \epsilon / 4) w(M(i)) \nonumber\\
\leq & \frac{3\epsilon}{8\maxweight} w(M(i))
\end{align*}
The bounds then follows from Lemma \ref{lem:stableweighted}
with $\epsilon' = 3 \epsilon / 8$.
\end{proof}

One other minor issue is that the next round need to access a previous
state of the graph.
As there is at most one round active at once, we can simply keep two
versions, the current one as well as the one from when the round was
initiated.
A conceptually simpler way to address this issue can also be
found in persistent data structures~\cite{DriscollSST89}, which
allows access to previous versions of the graph.

\end{appendix}

\end{document}